\documentclass[superscriptaddress,a4paper,nofootinbib,onecolumn,notitlepage,11pt,tightenlines]{revtex4-1}

\usepackage{amssymb,amsmath,latexsym}
\usepackage{amsthm}
\usepackage{dsfont}
\usepackage[usenames,dvipsnames]{color}
\usepackage{nicefrac}
\usepackage{graphicx}
\usepackage{pstool}
\usepackage{mathrsfs}

\usepackage[ colorlinks = true, 
             linkcolor = blue,
             urlcolor  = blue,
             citecolor = red,
             anchorcolor = green,
]{hyperref}

\newcommand{\tr}{\operatorname{tr}}
\newcommand{\id}{\mathds{1}}

\newcommand{\eps}{\varepsilon}
\newcommand{\rhot}{\tilde{\rho}}
\newcommand{\rhob}{\bar{\rho}}
\newcommand{\rhoh}{\hat{\rho}}
\newcommand{\taub}{\bar{\tau}}
\newcommand{\omegat}{\tilde{\omega}}

\newcommand{\cM}{\mathcal{M}}

\newcommand{\bX}{\mathsf{X}}
\newcommand{\bY}{\mathsf{Y}}
\newcommand{\bR}{\mathsf{R}}
\newcommand{\bS}{\mathsf{S}}
\newcommand{\bK}{\mathsf{K}}

\def\bra#1{\mathinner{\langle{#1}|}}
\def\ket#1{\mathinner{|{#1}\rangle}}

\def\braket#1#2{\mathinner{\langle{#1}|{#2}\rangle}}
\def\brakket#1#2#3{\mathinner{\langle{#1}|{#2}|{#3}\rangle}}
\def\abs#1{\left|{#1}\right|}
\def\absb#1{\big|{#1}\big|}

\def\proj#1{\ket{#1}\!\!\bra{#1}}

\theoremstyle{plain}
\newtheorem{lemma}{Lemma}
\newtheorem{proposition}[lemma]{Proposition}

\theoremstyle{plain}
\newtheorem{theorem}[lemma]{Theorem}
\theoremstyle{remark}
\newtheorem{definition}{Definition}

\begin{document}

\title{The Link between Entropic Uncertainty and Nonlocality}

\author{Marco Tomamichel}
\email{cqtmarco@nus.edu.sg}
\affiliation{Centre for Quantum Technologies, National University of Singapore, Singapore}
\affiliation{Institute for Theoretical Physics, ETH Zurich, 8093 Zurich, Switzerland}
\author{Esther H\"anggi}
\affiliation{Centre for Quantum Technologies, National University of Singapore, Singapore}

\begin{abstract}
Two of the most intriguing features of quantum physics are the 
\emph{uncertainty principle} and the occurrence of \emph{nonlocal correlations}.
The uncertainty principle states that there exist 
pairs of incompatible measurements on quantum systems such that their outcomes 
cannot both be predicted. On the other hand, nonlocal correlations of 
measurement outcomes at different locations cannot 
be explained by classical physics, but appear in the presence of entanglement.
Here, we show that these two fundamental quantum effects are
quantitatively related. 
Namely, we provide an entropic uncertainty relation for the outcomes of two binary measurements,
where the lower bound on the uncertainty is quantified in terms of the maximum 
Clauser-Horne-Shimony-Holt value that can be achieved with these measurements. 
We discuss applications of this uncertainty relation in quantum cryptography, in particular,
to certify quantum sources using untrusted devices.
\end{abstract}

\maketitle

\section{Introduction}


A remarkable characteristic of quantum physics is the
\emph{uncertainty principle}, as first 
described by Heisenberg~\cite{heisenberg27} and Robertson~\cite{robertson29}. It expresses the fact 
that there exist certain observable properties of a quantum system 
such that knowledge of one necessarily implies uncertainty 
about the other. 
In recent relations, starting with~\cite{deutsch83,maassen88}, the uncertainty of a measurement 
is often quantified in terms of entropies evaluated for the probability 
distribution over measurement outcomes induced by Born's rule.
Roughly speaking, if the distribution over the different measurement outcomes is close to uniform, 
the entropy is large and the uncertainty high; on the other hand, a peaked distribution leads to small entropy and low uncertainty. 
An entropic uncertainty relation provides a lower bound on the sum of the entropies
of two or more alternative measurements that is valid for all 
states of the quantum system prior to measurement. This bound is trivial for compatible measurements and
can generally be seen as a measure of ``incompatibility'' of the measurements.
We restrict the discussion to measurements with a finite number of different outcomes hereafter,
and point to a recent review of the topic by Wehner and Winter~\cite{wehner09}.

A prominent example of such an uncertainty relation is the one shown by Maassen and 
Uffink~\cite{maassen88}. It states that the Shannon entropy of the outcomes of two non-degenerate measurements, $X$ and $Y$, is lower bounded
by a function of their \emph{overlap}, $c$. Namely,
\begin{align}
  \label{eq:ucrnoside}
  H(X) + H(Y) \geq -\log_2 c, \qquad \textnormal{where} 
  \qquad c = \max_{i,j} \absb{\braket{\phi^i}{\psi^j}}^2 \,.
\end{align}
The overlap of the two measurements is a function of their eigenvectors, $\ket{\phi^i}$ and $\ket{\psi^j}$, respectively.
(We shall make this statement more formal in the following sections.)

In~\cite{berta10,tomamichel11,coles10,colbeck11}, entropic uncertainty relations have been extended to include 
the case where observers have access to a quantum memory, i.e.\ a quantum system that is correlated with the state prior to measurement. 
Note that an entangled observer can in principle perfectly predict the outcomes of both measurements appearing in Eq.~\eqref{eq:ucrnoside} by applying an appropriate measurement on his memory. Thus, Eq.~\eqref{eq:ucrnoside} is no longer valid when the Shannon entropies are replaced by von Neumann entropies conditioned on the observers memory.
(We refer to the discussion in~\cite{berta10} for more details.) This limitation can be overcome by introducing tripartite uncertainty relations, where one considers two separate quantum memories, $B$ (controlled by Bob) and $C$ (controlled by Charlie) and takes advantage of the monogamy of entanglement. Surprisingly, uncertainty relations of
a similar form as~\eqref{eq:ucrnoside} result, but now the uncertainty is formulated in terms of 
conditional  von Neumann entropies and reads~\cite{berta10}
\begin{align}
  \label{eq:ucr}
  H(X|B) + H(Y|C) \geq - \log_2 c \,.
\end{align}

This inequality can be interpreted as follows. If Bob can predict the outcome of the $X$ measurement with certainty (i.e., $H(X|B) = 0$), then Charlie necessarily has uncertainty about the outcome of the $Y$ measurement (i.e., $H(Y|C) > 0$) as long as the measurements are incompatible (i.e., $c < 1$).
Note also that~\eqref{eq:ucr} implies~\eqref{eq:ucrnoside} due to the strong sub-additivity 
of the von Neumann entropy~\cite{lieb73} and is, therefore, strictly stronger. 

In the context of cryptography, uncertainty of an eavesdropper
implies (partial) secrecy, and indeed entropic uncertainty 
relations have been employed to show cryptographic security~\cite{koashi06,damgaard07,tomamichel11,tomamichellim11,damgaard08}.
More generally, the usefulness of these uncertainty relations can be understood from the fact that the entropies on the lefthand side of~\eqref{eq:ucrnoside} and~\eqref{eq:ucr} characterize operational quantities in information theory, e.g.\ the asymptotic data compression rate~\cite{shannon48,devetak03}.  


Another phenomenon distinguishing quantum from classical physics is the occurrence of 
\emph{nonlocal correlations}. It has already been observed by Einstein, Podolsky and 
Rosen~\cite{epr35} that quantum mechanics predicts correlations between entangled, but 
spatially separated particles, which are stronger than one would intuitively expect.
Bell~\cite{bell64} later showed that these correlations 
cannot be explained by any classical local theory; hence, they are called nonlocal. 

Nonlocality can be quantified using so-called Bell inequalities~\cite{bell64}. 
A prominent example is the Clauser-Horne-Shimony-Holt (CHSH) inequality~\cite{CHSH}, 
which considers a bipartite setup where two separated parties, called Alice ($A$) and 
David ($D$), share a potentially entangled quantum state.
Both parties randomly choose one out of two binary measurements that they apply 
to their share of the quantum state.
We denote the outcomes of Alice's measurements by the random variables $X$ and $Y$ (as in the setup of the uncertainty relation) and David's outcomes by $R$ and $S$, depending on his choice of measurement.
The {CHSH} inequality states that, for any classically correlated state, 
it holds that $\beta \leq 2$, where
\begin{align}
  \label{eq:beta}
\beta &= 2\Pr[X = R] + 2\Pr[Y = R] + 2\Pr[X = S] + 2\Pr[Y \neq S] - 4
\end{align}
is called the \emph{{CHSH} value}. 
If $\beta > 2$, we call the correlation nonlocal, and quantum mechanics allows correlations that achieve up to $\beta_{\max} = 2\sqrt{2}$, which is called Tsirelson's bound~\cite{tsirelson80}.
(Nonlocal correlations can, for example, be realized using an entangled pair of spin-$\nicefrac{1}{2}$ particles, where the 
choice of measurement corresponds to a spin direction. 
However, we will not make any assumption about how the system is physically realized in the following.)

The remainder of this paper is structured as follows. Section~\ref{sc:rw} discusses related work. 
Section~\ref{sc:main} states the main results of our work, 
which provide a link between entropic uncertainty and nonlocality.  Finally,
Section~\ref{sc:app} sketches an
application of our results to self-testing sources of {Bennett-Brassard 84} states. 
The formal proofs
of the main results are deferred to the appendix.

\section{Related Work}
\label{sc:rw}

The main result of this paper is a quantitative relation between entropic uncertainty and nonlocality.
The fact that the incompatibility of local measurements and nonlocality are related in some way
is folklore knowledge and follows, for example, from the work of Tsirelson~\cite{tsirelson80}.
For the case when the systems are restricted to qubits, a bound on the maximal CHSH value in
terms of the angle between local measurements has been derived by Seevink and Uffink~\cite{seevinck07}.
The analytical form of Relation~\eqref{eq:overlapupperbound}
has been conjectured by Horodecki~\cite{horodecki-pc11} and derived independently by 
Lim~\cite{lim-pc12} for the case of single qubit systems.
Mayers and Yao have shown that in order to reach the maximal 
CHSH value allowed by quantum physics, the state and measurements essentially 
need to be (equivalent to) a fully entangled state and optimal {CHSH} measurements even when they 
are embedded in higher dimensions~\cite{my98,my04}. They also employed this result in
quantum cryptography, where they used it to construct self-testing sources.

We improve these results by providing an exact
analytical relation that characterizes all allowed combinations of local overlap and CHSH value.
In particular, our result is independent of the system dimension and
the quantum state under consideration. 
Furthermore, the overlap\,|\,in contrast to other measures of incompatibility based on the 
commutator of the observables or the angle between measurements that have been 
investigated previously\,|\,attains operational
meaning in quantum information theory through the entropic uncertainty relations.
Following Mayers and Yao, we also sketch an application our result to self-testing sources.

On a related topic, Oppenheim and Wehner~\cite{oppenheim10}\,|\,for a class of generalized physical theories that includes quantum mechanics and classical theory\,|\,showed that the presence of uncertainty, via steering, directly limits the maximally achievable nonlocality. 
Our result can be seen as complementary to theirs, as we show that in order 
to achieve a certain nonlocality, at least some specific amount of uncertainty is necessary. 

Device-independent quantum key distribution~\cite{magniez06,acin07,
mckague10,haenggirenner10,masanes11} 
and randomness 
generation~\cite{colbeck10} usually 
bases security on a relation between nonlocality and the randomness of the 
outcomes relative to some (quantum) adversary. Our result allows
to split the security analysis of these protocols into two parts: the nonlocality
of the measured correlations first gives a bound on the uncertainty of
local measurement outcomes, which in turn can be used to ensure security. 
The two parts can be analyzed independently and thus our methods can be used to
simplify such an analysis and, potentially, reduce the required assumptions.

\section{Main Results}
\label{sc:main}

In order to present our main results, we employ the density operator formalism of quantum mechanics in finite dimensions and use standard notation that we quickly summarize here.

\subsection{Notation}

A \emph{quantum state} is represented by a positive semidefinite operator with unit trace acting on a finite-dimensional Hilbert space. We consider states shared between different locations, which are described as operators acting on the tensor product of the respective local spaces. 
For example, we denote by $\rho_{AB}$ a state shared between
locations $A$ and $B$ and by $\rho_B = \tr_A(\rho_{AB})$ its marginal state on $B$, where $\tr_A$ is the partial trace over $A$. 

A quantum measurement can be most generally described by a \emph{positive operator-valued measure} (POVM).
The measure induces a \emph{completely positive trace-preserving map} (CPTPM) that maps states on $A$ to a classical register that contains the measurement outcome. Within the quantum formalism, a classical register (or random variable) is described by a Hilbert space with a fixed basis and states that are 
diagonal in this basis.
For example, let $\bX = \{ M_A^x \}$ be a measurement with discrete outcomes on $A$, i.e.\ a set indexed by $x$ of positive semidefinite operators $M_A^x$ on $A$ satisfying 
$\sum_x M_A^x = \id_A$, where $\id_A$ is the identity operator on $A$. 
The corresponding measurement map, $\cM_{\bX}$ from $A$ to the register $X$, thus produces states of the form
\begin{align*}
  \cM_{\bX} : \rho_{AB} \mapsto \rho_{XB} = \sum_x \proj{x}_X \otimes 
    \tr_A \big( ( M_A^x \otimes \id_B ) \rho_{AB} \big) = \sum_x p_x \proj{x}_X \otimes \rho_B^x \,,
\end{align*}
where $p_x = \tr(M_A^x \rho_A)$ is the probability with which outcome $x$ occurs, $\rho_B^x = \frac{1}{p_x} \tr_A
\big( M_A^x\, \rho_{AB} \big)$ is the state of $B$ conditioned on the event that $x$ 
was measured and $\proj{x}_X$ is the projector onto an element of a fixed orthonormal basis $\{ \ket{x} \}$ 
of $X$. (Note that we often omit writing the identity operator when it is clearly implied by context.)
We call a measurement \emph{projective} if the operators $M_A^x$ are projectors, i.e.\ if $M_A^x M_A^x = M_A^x$ for all $x$.

We also use the fact that non-projective measurements can seen as projective measurements of an enlarged quantum system. More precisely, a \emph{dilation} of a measurement $\bX = \{ M_A^x \}$ consists of an embedding $U: A \to A'$ that embeds $A$ into a larger space $A'$ and a measurement $\bX' = \{ M_{A'}^x \}$ on $A'$ such that $U^{\dagger} M_{A'}^x U = M_A^x$ for all $x$. The latter condition ensures that, for every state $\rho_{AB}$, we have $\rho_{XB} = \cM_{\bX}[\rho_{AB}] = \cM_{\bX'}[U \rho_{AB} U^\dagger]$, i.e.\ the post measurement states of the two measurements are equal. Moreover, Neumark's dilation theorem~\cite{neumark43} ensures that if $A'$ is chosen sufficiently large, there always exists a dilation such that $\bX'$ is projective.

We employ the operator norm $\| \cdot \|$, which evaluates to the largest eigenvalue for Hermitian operators.
Moreover, we define the conditional von Neumann entropy, $H(A|B)_{\rho} := H(AB)_{\rho} - H(B)_{\rho}$, where 
$H(A)_{\rho} := - \tr ( \rho_A \log_2 \rho_A )$. Note that for the above example $H(X)_{\rho}$ reduces to the 
Shannon entropy of the probability distribution $p_x$ induced by the measurement and that $H(X|B)_{\rho} \leq 
H(X)_{\rho}$ due to the strong sub-additivity of the von Neumann entropy~\cite{lieb73}.

This formalism allows us to restate the uncertainty relation~\eqref{eq:ucr} in its full generality~\cite{krishna01,tomamichel11,coles10}.\\
Given any tripartite quantum state $\rho_{ABC}$ and two measurements 
$\bX = \{ M_A^x \}$ and $\bY = \{ N_A^y \}$ on $A$, 
the post measurement states 
$\rho_{XB} = \cM_{\bX}[\rho_{AB}]$ and $\rho_{YC} = \cM_{\bY}[\rho_{AC}]$ satisfy
\begin{align}
  H(X|B)_{\rho} + H(Y|C)_{\rho} \geq - \log_2 c(\bX, \bY) \,, \quad \textrm{where} 
  \quad c(\bX, \bY) := \max_{x,y} \Big\| \sqrt{M_A^x}\, N_A^y \sqrt{M_A^x} \Big\| . \label{eq:ucrfull}
\end{align}

This relation gives a bound on the uncertainty in terms of the overlap which is a function of the two measurements but independent of the quantum state of the system prior to measurement. Note that $c(\bX, \bY)$ reduces
to the expression in~\eqref{eq:ucrnoside} in the case of non-degenerate projective measurements.

\subsection{Generalized Uncertainty Relations}

While the overlap, and thus the uncertainty, can 
be calculated from the POVM elements associated with the two 
measurements alone, it cannot be tested experimentally. 
Hence, in practice, determining the uncertainty a measurement produces 
requires a precise theoretical model of the measurement devices used and any deviation of the physical implementation from this theoretical model may
lead to an overestimation of the produced uncertainty.
Specifically, this is of critical importance in quantum cryptography, where uncertainty of one observer ensures security for the others, and an overestimation of this uncertainty directly leads to a
security loophole.


In this work, we will thus introduce a variation of the overlap, the 
\emph{effective overlap}, which can be tested experimentally in an
important special case as we will see below.
The definition of the effective overlap is motivated by the following two observations. 
\begin{itemize}

\item The entropies on the left-hand side of the uncertainty relation~\eqref{eq:ucrfull} are evaluated for the post measurement states
$\rho_{XB}$ and $\rho_{YC}$ that result from measuring $\bX$ and $\bY$ on $\rho_{ABC}$,
respectively.
However, these post measurement states can generally also be constructed in other ways and it is evident that the right-hand side of~\eqref{eq:ucrfull}
can thus be maximized over all pairs of measurements that achieve the post measurement states 
$\rho_{XB}$ and $\rho_{YC}$.
A generic construction of such measurements is given by any pair of joint dilations $\{ U, \bX' \}$ and $\{ U, \bY' \}$ of $\bX$ and $\bY$ based on the same embedding $U: A \to A'$. The post measurement states can now
alternatively
be constructed as $\rho_{XB} = \cM_{\bX'}[U \rho_{AB} U^\dagger]$ and $\rho_{YC} = \cM_{\bY'}[U \rho_{AC} 
U^\dagger]$ and the right-hand side of~\eqref{eq:ucrfull} can be evaluated either for $c(\bX, \bY)$ or for
$c(\bX', \bY')$.

\item Moreover, any projective 
measurement on $A$\,|\,let us denote it by $\bK = \{ P_A^k \}$\,|\,can be used to slice the state 
into orthogonal parts before the actual measurements are applied. This results in an intermediate 
state of the form $\sum_k P_A^k\, \rho_{ABC} P_A^k$. Moreover, if this extra measurement commutes
with both $\bX$ and $\bY$ on the support of $\rho_A$, the respective post measurement states with and without slicing are indistinguishable, i.e.\ we have $\rho_{XB} = \cM_{\bX}[\rho_{AB}] = 
\cM_{\bX}\big[\sum_k P_A^k\, \rho_{AB} P_A^k\big]$ and $\rho_{YC} = \cM_{\bY}[\rho_{AC}] = 
\cM_{\bY}\big[\sum_k P_A^k\, \rho_{AC} P_A^k\big]$. We will see in the following that
the overlap of the measurements $\bX$ and $\bY$ on the sliced state is given by the
average overlap evaluated for the individual slices.

\end{itemize}

We combine these two observations to
define the \emph{effective overlap} as a function of a \emph{measurement setup}, which consists of two measurements 
and the marginal state $\rho_A$ on $A$ that will be measured. 

\begin{definition}
  Let $\rho_A$ be a quantum state and let
  $\bX = \{ M_A^x \}$ and $\bY = \{ N_A^y \}$ be two measurements on $A$.
  The effective overlap of the measurement setup 
  $\{ \rho_A, \bX, \bY \}$ is defined as
  \begin{align*}
    c^*(\rho_A, \bX, \bY) := \inf_{ U,\, \bX', \bY', \bK' }
    \Bigg\{ \sum_k \tr\,( P_{A'}^k U \rho_{A}^{\phantom{k}} U^{\dagger})\, 
     \max_x \bigg\| \sum_y P_{A'}^k N_{A'}^y P_{A'}^k \cdot
     P_{A'}^k M_{A'}^{x\phantom{y}\!\!} P_{A'}^k \cdot P_{A'}^k N_{A'}^y P_{A'}^k \bigg\| \Bigg\}
  \end{align*}
  where the infimum is taken over all embeddings $U$ from $A$ to an auxiliary space $A'$, all measurements $\bX' = \{ M_{A'}^x \}$ and $\bY' = \{ N_{A'}^y \}$ on $A'$, and all projective measurements $\bK' = \{ P_{A'}^k \}$ on $A'$ such that
  $\sum_k U^{\dagger} P_{A'}^k M_{A'}^x P_{A'}^k U = M_A^x$ and 
  $\sum_k U^{\dagger} P_{A'}^k N_{A'}^y P_{A'}^k U = N_A^y$ for all $x$ and $y$.
\end{definition} 

Note that while evaluating the effective overlap for a general measurement setup might be intractable, it is often easy to find upper bounds on it. To see this, consider the following example, where the effective overlap leads to a tighter
characterization of the uncertainty. 

We apply one of two projective
measurements, either in the basis $\{ \ket{0}, \ket{1}, \ket{\perp} \}$ or in
the basis $\{\ket{+}, \ket{-}, \ket{\perp} \}$, where $\ket{\pm} = (\ket{0} \pm \ket{1})/\sqrt{2}$.
These measurements are applied 
on a state $\rho$ which has the property that
`$\perp$' is measured with probability at most $\eps$. 
The uncertainty relation~\eqref{eq:ucrfull} gives a trivial bound as the overlap 
of the two bases is $c = 1$. Still, our
intuitive understanding is that the uncertainty about
the measurement outcome is high as long as $\eps$ is small. In fact, it is easy to verify that the 
effective overlap of this setup
satisfies $c^* \leq (1 - \eps) \frac{1}{2} + \eps$ and thus captures this intuition.
(This formula can be interpreted as follows: with probability $1-\eps$ we are in the subspace spanned by $\ket{0}$ and $\ket{1}$, where the overlap is $\frac{1}{2}$, and with probability $\eps$ we measure $\perp$ and have maximal overlap.)

Our first result is a generalization of the uncertainty
relations~\eqref{eq:ucrnoside} and~\eqref{eq:ucr}. 
We show that these relations still hold when the overlap is replaced 
by the effective overlap. 

\begin{theorem}\label{th:ucr}
  Let $\rho_{ABC}$ be a tripartite quantum state and 
  $\bX = \{ M_A^x \}$ and $\bY = \{ N_A^y \}$ two measurements on $A$. Then,
  the states $\rho_{XB} = \cM_{\bX}[\rho_{AB}]$ and 
  $\rho_{YC} = \cM_{\bY}[\rho_{AC}]$ satisfy
  \begin{align*} 
    H(X)_{\rho} + H(Y)_{\rho} &\geq H(X|B)_{\rho} + H(Y|C)_{\rho} \geq -\log_2 {c^*(\rho_A, \bX, \bY)}\,. 
  \end{align*}
\end{theorem}

The proof of this theorem employes the \emph{smooth entropy framework}~\cite{renner05,tomamichel09,mythesis}, which
has already found many applications in quantum cryptography and non-asymptotic information theory.
In the process, we also generalize an entropic uncertainty relation for smooth 
entropies~\cite{tomamichel11}.

Let us thus explain in more detail why an uncertainty relation in 
terms of smooth min-and max-entropy is desirable. 
The von Neumann entropy used above (and its classical analogue, the 
Shannon entropy) characterizes  
information theoretic tasks in the asymptotic limit 
of many independent repetitions. In practice, one can 
neither perform an infinite number of repetitions 
of an experiment, nor are the different runs usually 
independent of each other. In the setting where we would 
like to characterize the resources related to a 
task which is repeated only once, called the \emph{one-shot} setting, 
the smooth min- and max-entropies often take the role
of von Neumann entropy. In order for them to be applicable to the analysis of 
realistic protocols it is therefore 
crucial to develop uncertainty relations in terms 
of smooth entropies. 

The smooth entropies can be interpreted as operational 
quantities in the following sense. On the one hand, the smooth 
min-entropy, $H_{\min}^\eps(X|B)$, quantifies the maximal 
number of uniformly random bits, independent of quantum 
side information $B$, that can be extracted from $X$~\cite{renner05,tomamichel10}. 
This quantity is of particular importance in cryptography, 
were the task often involves extracting randomness that 
is secret from a quantum adversary. On the other hand, 
the smooth max-entropy, $H_{\max}^{\eps}(Y|C)$, quantifies 
the minimum number of additional bits of information about $Y$ that
are needed to reconstruct $Y$ from a quantum memory $C$~\cite{renesrenner10}.
In both cases, the \emph{smoothing parameter}, $\eps$, 
ensures the quality of the resulting state, i.e.\ it has 
to be indistinguishable from a perfect output up to 
probability $\eps$.

The following relation is thus of independent interest and shows that the uncertainty relation
for smooth entropies in~\cite{tomamichel11} also holds for the effective overlap.

\begin{theorem}\label{th:ucrmin}
  Let $\rho_{ABC}$ be a tripartite quantum state, $\eps \geq 0$, $\bar{\eps} > 0$ and let 
  $\bX = \{ M_A^x \}$ and $\bY = \{ N_A^y \}$ two POVMs on $A$. Then,
  the states $\rho_{XB} = \cM_{\bX}[\rho_{AB}]$ and 
  $\rho_{YC} = \cM_{\bY}[\rho_{AC}]$ and the smooth min- and max-entropies as 
  defined in Appendix~\ref{sc:proofucr} satisfy
\begin{align*}
  H_{\min}^{\eps+2\bar{\eps}}(X|B)_{\rho} + H_{\max}^{\eps}(Y|C)_{\rho} \geq - \log_2 c^*(\rho_A, \bX, \bY) - \log_2\, (2/\bar{\eps}^2) \,.
\end{align*}
\end{theorem}

The relation in the above form directly leads to a formal 
security proof of quantum key distribution (QKD) against general adversaries 
while at the same time making it more robust against device imperfections, 
in analogy with~\cite{tomamichel11,tomamichellim11}. 
To see how this works, consider the entanglement based version of the
Bennett-Brassard 1984 protocol~\cite{bb84,bennett92} 
and $n$ measurements in the computational and diagonal basis such that $- \log_2 c^* = n$.
The uncertainty relation is now applied to the situation where Alice and Bob would 
like to agree on a key, while Charlie takes the role of the eavesdropper. 
Using the operational meaning of the smooth 
entropies as described above, the uncertainty relation states 
that the number of secret bits extractable from a raw 
string $Y^n$ is given by $n$ minus the number of additional 
bits from Alice required for Bob to correct phase errors (i.e.\ the errors in $X^n$).
The latter number, however, can be inferred by Alice and Bob from experimental data, 
and thus the security of the extracted key can be ensured by them without making
any assumptions about the eavesdropper's attack.

The detailed proofs of Theorem~\ref{th:ucr} and~\ref{th:ucrmin} can be found in Appendix~\ref{sc:proofucr}.

\subsection{Relation between Overlap and Nonlocality}

We now consider four POVM measurements with \emph{binary} outcomes, 
$\bX$ and $\bY$ on Alice's side as well as $\bR$ and $\bS$ on David's side.
We first define the \emph{{CHSH} value} of  a \emph{bipartite measurement setup}.

\begin{definition}
  Let $\rho_{AD}$ be a bipartite state and let $\bX = \{ M_A^0, M_A^1 \}$,
  $\bY = \{ N_A^0, N_A^1 \}$ be measurements on $A$ and $\bR = \{ R_D^0, R_D^1 \}$,
  $\bS = \{ S_D^0, S_D^1 \}$ be measurements on $D$. Then, the {CHSH} value of 
  the bipartite measurement setup $\{\rho_{AD}, \bX, \bY, \bR, \bS\}$ is defined as
  \begin{align*}
    \beta(\rho_{AD}, \bX, \bY, \bR, \bS) := 2 
    \tr \bigg( \sum_{i = 0}^{1} \big(M_A^i \otimes (R_D^i + S_D^i) + N_A^i 
      \otimes (R_D^i + S_D^{1-i}) \big)\, \rho_{AD} \bigg) - 4 \,.
  \end{align*}
\end{definition}
\noindent Note that the trace term corresponds to $\Pr[X = R] + \Pr[Y = R] + \Pr[X = S] + \Pr[Y \neq S]$ in~\eqref{eq:beta} evaluated for the state $\rho_{AD}$ and the four specified POVMs.

The main result of this paper shows a relation between the 
effective overlap of Alice's measurement setup and $\beta$, the maximal {CHSH} value
that can be reached between Alice and an arbitrary additional party, David, with the same measurement
setup on Alice's side. (Alice's measurement setup is given by the marginal state on $A$ as well as
the two possible POVMs she can choose from.)

\begin{theorem}
  \label{th:maxoverlap}
  Let $\rho_A$ be a state and let $\bX$, $\bY$ be binary measurements such that
  $c^* = c^*( \rho_A, \bX, \bY \}$. Then, for any $\rho_{AD}$
  with $\tr_{D} (\rho_{AD}) = \rho_A$ and 
  any two binary measurements $\bR$, $\bS$ on $D$, we have
  \begin{align}
    \beta(\rho_{AD}, \bX, \bY, \bR, \bS) 
      \leq 2 \big(\sqrt{c^*}+\sqrt{1-c^*}\big) . \label{eq:overlapupperbound}
  \end{align}
  Conversely, for any bipartite state $\rho_{AD}$ and any binary measurements $\bX$, $\bY$ 
  on $A$ and $\bR$, $\bS$ on $D$ such that
  $\beta = \beta(\rho_{AD}, \bX, \bY, \bR, \bS)$, we have
  \begin{align}
    \label{eq:overlaplowerbound} 
    c^*(\rho_A, \bX, \bY) \leq \frac{1}{2} + \frac{\beta}{8} \sqrt{8-\beta^2}\;.
  \end{align}
\end{theorem}

This bound is depicted in Figure~\ref{fig:overlap} and implies as a special case 
that any state and measurement on Alice's part which can give rise to
nonlocal correlations 
(i.e., $\beta > 2$), must have effective overlap $c^* < 1$. Furthermore, in order
to reach a 
{CHSH} value close to Tsirelson's bound (i.e., $\beta \approx 2\sqrt{2}$),
the measurement on $A$ must have almost minimal overlap $c^* \approx 1/2$. 

\begin{figure}[ht!]
\psfragfig{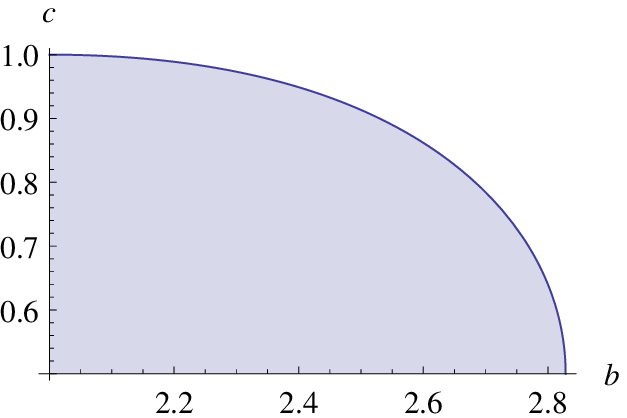}{%
    \psfrag{b}{$\beta$}
    \psfrag{c}{$c^*$}
}
\hspace{1cm}
\psfragfig*{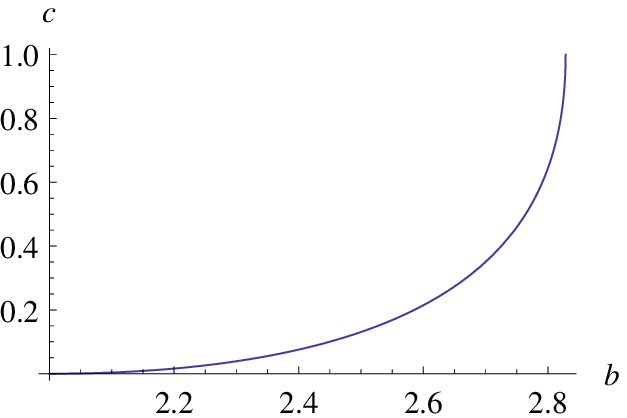}{%
    \psfrag{b}{$\beta$}
    \psfrag{c}{$q$}
}
\caption{\label{fig:overlap} The relation between local overlap and {CHSH} value. Due to our bound~\eqref{eq:overlaplowerbound}, combinations of $\beta$ and $c^*$ outside the filled region
in the left figure are impossible. The right figure shows the guaranteed uncertainty, $q = -\log_2 c^*$,
as a function of $\beta$.}
\end{figure}

Theorem~\ref{th:maxoverlap} in particular implies that if Alice and David can experimentally 
verify that the {CHSH} violation of their bipartite setup exceeds some fixed value $\beta$,
then the effective overlap of both Alice's and David's local measurements is upper bounded 
by~\eqref{eq:overlaplowerbound}. 

Finally, Equation~\eqref{eq:overlaplowerbound}, together with Theorem~\ref{th:ucr}, 
directly implies an 
uncertainty relation with quantum side information where the lower bound is stated 
in terms of the {CHSH} value the measurement setup can reach. 
This \emph{device-independent uncertainty relation} is stated 
only in terms of 
quantities which have an operational meaning. We have
\begin{align*}
    H(X)_{\rho} + H(Y)_{\rho} \geq 
  H(X|B)_{\rho} + H(Y|C)_{\rho} \geq 1-\log_2 \left(1+\frac{\beta}{4}\sqrt{8-\beta^2} \right) ,
\end{align*}
where $\beta$ is the {CHSH} value between $A$ and $D$, 
resulting from measuring any state $\rho_{AD}$ 
with $\tr_T (\rho_{AD}) = \rho_A$ using measurements $\bX$ and $\bY$ on $A$ 
and arbitrary measurements on $D$. The right-hand side of this inequality, i.e.\ the guaranteed uncertainty, is also depicted in Figure~\ref{fig:overlap}.

This implies, for example, that if Bob's uncertainty about Alice's outcome is low, 
but the {CHSH} value between Alice and Bob (who takes the role of David in this example) is high, 
then Charlie's uncertainty about the outcome of the other measurement must necessarily be 
high. Alice and Bob can therefore infer whether Charlie has high 
entropy from their correlations alone.

We want to stress again that previous uncertainty relations were stated in terms of the overlap,
which can only be determined if the exact specification of Alice's measurement devices
is known. Our uncertainty relation, on the other hand, depends \emph{only} on the observable quantity $\beta$ and is independent of the details of the theoretical model used to describe the quantum 
systems and measurements. This includes, in particular, the dimension of the Hilbert space they act on. 

We refer to Appendix~\ref{sc:overlap} for the proof of Theorem~\ref{th:maxoverlap}.

\section{Application: Certification of {BB84}-Sources}
\label{sc:app}

Theorem~\ref{th:maxoverlap} can be used 
to test the effective overlap in a device-independent way, i.e., 
where the test equipment does not need to be trusted. 
Such a test could, for example, be used by 
manufacturers  
to certify the quality of a source 
creating {BB84}-states~\cite{bb84} and to 
proof to a skeptical audience that their 
devices fulfill the desired specifications. 
Sources of {BB84}-states are 
widely used in quantum cryptography,
including quantum key distribution
and bit commitment or oblivious transfer secure in the 
bounded/noisy storage model~\cite{damgaard08,koenig09}.  
Moreover, recent security proofs for quantum key 
distribution~\cite{berta10,tomamichel11,tomamichellim11} are based on uncertainty relations of the form~\eqref{eq:ucr}. The overlap of the source enters there as the crucial parameter determining the secrecy of the resulting key\,|\,in particular, there is no need to do 
tomography of the produced states. For this reason, the overlap can be regarded as 
the key parameter quantifying the quality of sources of {BB84}-states. 

\begin{figure}[h]
\centering
\includegraphics[width=10cm]{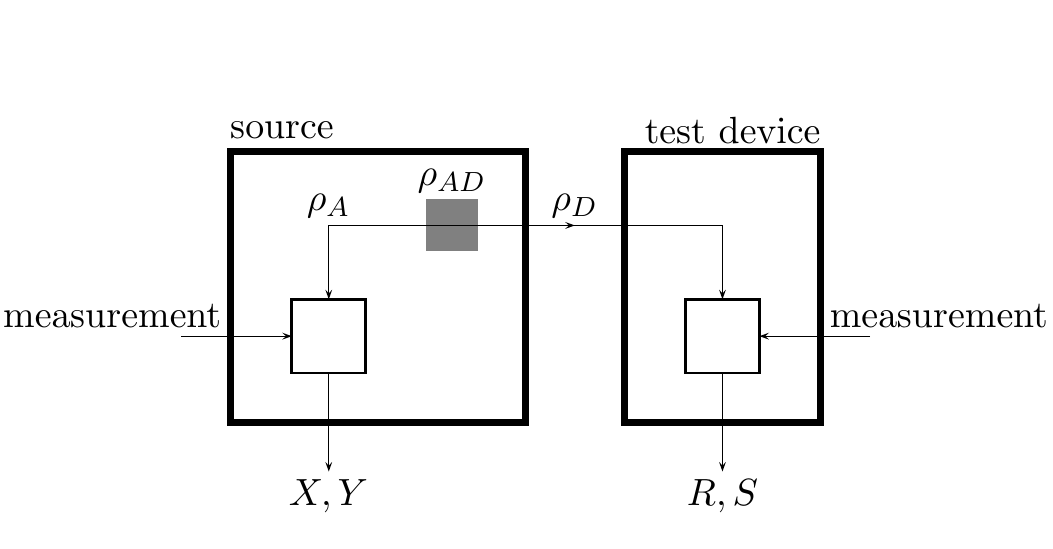}
\caption{\label{fig:source} Certification of entanglement-based sources of {BB84}-states.}
\end{figure}

Consider a (potentially imperfect) source that creates {BB84}-states in the following way (see Figure~\ref{fig:source}). First, it produces two entangled particles in a state $\rho_{AD}$, e.g.~through parametric 
down-conversion~\cite{shih88,kiess93}. Then, it emits one part, $D$, of the entangled quantum state and measures the other part, $A$, using one of two different measurements chosen at random. 
Denote the binary measurement outcome by $X$ or $Y$ depending on the input. The input of the source thus corresponds to the choice of basis 
for the {BB84}-states, and, together with the output, defines which of the $4$ states was actually 
prepared. 
Sources of this type are the subject of recent research, e.g.~they are used as heralded 
single photon sources~\cite{pittman05,xiang10} and have applications in (device-independent) 
quantum cryptography~\cite{gisin10,curty11,pitkanen11}. 

A source which repeatedly and independently prepares states in this way  
can be certified by a test device 
which measures the emitted particle $D$ 
in one of two bases chosen at random  and outputs the 
measurement result, denoted by $R$ or $S$ depending on the input. 
The effective overlap of the source can then be estimated from the fraction $p=k/N$ of times the {CHSH} condition is satisfied (i.e., either $X=R$, $X=S$, $Y=R$ or $Y \neq S$), as
\begin{align*}
c^* \approx \frac{1}{2} + 2\left(2p-1\right) \sqrt{\frac{1}{2}-\left(2p-1\right)^2} \,.
\end{align*}
The precise evaluation of the statistics is straightforward but beyond the scope of this work.

\section{Conclusion}
\label{sc:conc}

We have found a novel relation between the local uncertainty of measurement outcomes (expressed in terms of the von Neumann or smooth min- and max-entropy) and nonlocality (expressed in terms of the CHSH value). This relation provides analytical bounds on the unpredictability of local measurement outcomes and opens a new avenue for device-independent quantum cryptography. Namely, it enhances the cryptographic applications of the entropic uncertainty relations since the crucial parameter, the effective overlap, can be tested experimentally.

Our result is limited to the CHSH Bell test and thus only considers binary measurements. Hence, a note of caution is advised here.
The {CHSH} value is naturally determined using measurements with binary outcomes. In practical experimental situations, however, often a third result occurs indicating that the measurement was unsuccessful.
There are different ways to deal with this situation. If we randomly or deterministically assign
one of the binary outcomes to this event, we stay in the framework of binary POVMs and the calculated $\beta$ indeed gives an upper bound on the effective overlap. If these unwanted results are simply discarded, however, we open the so-called post-selection loophole and our result does not apply without further analysis.

It remains an open question whether other Bell tests can be employed to bound the effective overlap of measurements with more
than two outcomes.

\section*{Acknowledgements}

We thank Micha\l{} Horodecki, Charles Ci Wen Lim, Corsin Pfister, Renato Renner, 
L{\'i}dia del Rio, Stephanie Wehner, 
Severin Winkler for helpful comments and discussions. 
EH and MT acknowledge support from the National Research Foundation (Singapore), and the
Ministry of Education (Singapore). MT is also supported by the Swiss National
Science Foundation through the National Centre of Competence in Research
`Quantum Science and Technology'.

\bibliographystyle{arxiv}
\bibliography{library}

\newpage
\appendix

\section{Proof of Generalized Uncertainty Relations}
\label{sc:proofucr}

\subsection{Preliminaries}

For the proof, we need two conditional entropies that are generalizations of the von Neumann entropy, the smooth min- and max-entropy. In order to define these, we first need to introduce the concept of sub-normalized quantum states and the purified distance. A sub-normalized quantum state is a positive semidefinite operator $\rho$ with $0 < \tr(\rho) \leq 1$ on a Hilbert space. 

The \emph{purified distance}~\cite{tomamichel09} between two sub-normalized quantum states, $\rho$ and $\tau$, is given by $P(\rho, \tau) := \sqrt{1 - F^2(\rho, \tau)}$, where $F(\rho, \tau) := \tr{\abs{\sqrt{\rho}\sqrt{\sigma}}} + \sqrt{(1-\tr{\rho})(1-\tr{\tau})}$ is the generalized fidelity. We say that the two states are $\eps$-close, denoted $\rho 
\approx^\eps \tau$, if and only if $P(\rho, \tau) \leq \eps$. The purified distance is a metric and has various important properties, e.g.\ $\rho \approx^\eps \tau \implies \mathcal{E}(\rho) \approx^\eps \mathcal{E}(\tau)$ for all trace non-increasing completely positive maps $\mathcal{E}$~\cite{tomamichel09}. 

Furthermore, due to Uhlmann's theorem, there exists an extension $\tau_{AB}$ of $\tau_A = \tr_B(\tau_{AB})$ such that $P(\rho_{AB}, \tau_{AB}) = P(\rho_A, \tau_A)$ for any bipartite state $\rho_{AB}$.
This state can be constructed (see~\cite{dupuis10}, Lemma B.2) and has the form 
\begin{align}
  \tau_{AB} = (X_A \otimes \id_B) \rho_{AB} (X_A^\dagger \otimes \id_B) \label{eq:pd-ext}
\end{align}
for some linear operator $X_A$ on $A$.
We use `$\succeq$' to denote the positive semidefinite partial order on Hermitian matrices, i.e.\
$A \succeq B$ if and only if $A - B$ is positive semidefinite.

\begin{definition}
  Let $\rho_{AB}$ be a sub-normalized state. The \emph{min-entropy} of $A$ given 
  $B$ is~\cite{renner05}
  \begin{align}
    \nonumber H_{\min}(A|B)_{\rho} &:=
    \max_{\sigma_B} \, \sup  \{ \lambda\in \bR : \rho_{AB} \preceq 2^{-\lambda}
    \mathds{1}_A\otimes \sigma_B  \}\,,
  \end{align}
  where the maximization is over all states $\sigma_B$ 
  on $B$.   For $\eps \geq 0$, the \emph{$\eps$-smooth min-entropy} and the
  \emph{$\eps$-smooth max-entropy} of $A$ given $B$ are defined as~\cite{koenig08,tomamichel09}
  \begin{align}
    \nonumber H^\eps_{\min}(A|B)_{\rho}
     := \max_{\rhot}\ H_{\min}(A|B)_{\rhot} 
    \qquad \text{and} \qquad
    H^\eps_{\max}(A|B)_{\rho}
     := - H^\eps_{\min}(A|C)_{\rho}
  \end{align}
  where the optimization is over all sub-normalized 
  states $\tilde{\rho}_{AB} \approx^\varepsilon \rho_{AB}$
  and $\rho_{ABC}$ is an arbitrary purification of $\rho_{AB}$.
\end{definition}
We note that in the limit of many independent copies of a quantum state, 
$\tau_{A^nB^n} = \rho_{AB}^{\otimes n}$, the smooth entropies converge to the von Neumann entropy~\cite{tomamichel08,mythesis}. For any $0 < \eps < 1$,
\begin{align}
  \lim_{n \to \infty} \frac{1}{n} 
    H_{\min}^\eps(A^n|B^n)_{\tau} = 
  \lim_{n \to \infty} \frac{1}{n} 
    H_{\max}^{\eps}(A^n|B^n)_{\tau} = H(A|B)_{\rho} \label{eq:aep}\,.
\end{align}

The smooth entropies satisfy various data-processing inequalities, in particular, for every CPTPM 
$\mathcal{E}$ from $B$ to $B'$, we have~\cite{tomamichel09}
\begin{align}
  H^\eps_{\min}(A|B)_{\rho} \leq H^\eps_{\min}(A|B')_{\tau} \quad \textrm{and} \quad 
  H^\eps_{\max}(A|B)_{\rho} \leq H^\eps_{\max}(A|B')_{\tau} \quad \textrm{for}\ \ \tau_{AB'} = \mathcal{E}[\rho_{AB}] \,. \label{eq:data-proc}
\end{align}

Finally, we need the following result. (See also~\cite{mythesis} for a slightly more general statement.)

\begin{lemma}
  \label{lm:pt-norm-bound}
  Let $M_{AB} \succeq 0$ and and let $\{ E_A^k \}_k$ be a set of linear operators on $A$.
  Then, 
  \begin{align}
    \tr_{A} \Big(\sum_k (E_A^k \otimes \id_B) M_{AB} ( E_A^k{}^{\dagger} \otimes \id_B) \Big) \preceq 
    \Big\| \sum_k E_A^k{}^{\dagger} E_A^k \Big\| \tr_A(M_{AB}) \,.
  \end{align}
\end{lemma}

\begin{proof}
  Due to the linearity and cyclicity of the partial trace, we
  have
  \begin{align*}
    \tr_A \big(\sum_k (E_A^k \otimes \id_B) M_{AB} ( E_A^k{}^{\dagger} \otimes \id_B) \big) 
    = \tr_A \Big( \sum_k ( E_A^k{}^\dagger E_A^k \otimes \id_B)  M_{AB} \Big) 
  \end{align*}
  We introduce the operator $R_{A} = \id_A \big\| \sum_k E_A^k{}^{\dagger} E_A^k \big\| - 
  \sum_k E_A^k{}^{\dagger} E_A^k \succeq 0$. We note that $\tr_A \big( (\sqrt{R_A} \otimes \id_B) M_{AB} 
  (\sqrt{R_A} \otimes \id_B) \big) \succeq 0$ 
  and, thus,
  \begin{align*}
   \tr_A \Big( \sum_k ( E_B^k{}^\dagger E_B^k \otimes \id_B)  M_{AB} \Big) 
   &\preceq \tr_A \Big( \sum_k ( ( E_B^k{}^\dagger E_B^k + R_A ) \otimes \id_B)  M_{AB} \Big) \\
   &= \Big\| \sum_k E_A^k{}^{\dagger} E_A^k \Big\| \tr_A(M_{AB}) \,.
    \qedhere
  \end{align*}
\end{proof}

\subsection{Smooth Relative Entropy}

Our proof relies heavily on the following auxiliary quantity, related to the relative 
max-entropy~\cite{datta08}, 
$h_{\min}(\rho \| \sigma) := \sup \{ \lambda \in \bR : \rho \preceq 2^{-\lambda} \sigma \}$.
It is easy to see that this quantity is monotonic under the application of a quantum map, i.e.\ $h_{\min}(\mathcal{E}[\rho]\, \|\, \mathcal{E}[\sigma]) \geq
h_{\min}(\rho \| \sigma)$ for all CPTPMs $\mathcal{E}$.

The following lemma relates the min-entropy and the relative entropy of the state and its marginal. 
(We refer to~\cite{tomamichel10} for a proof.)

\begin{lemma}
  \label{lm:rel-min-bound}
  Let $\eps > 0$ and $\rho_{ABC}$ a pure quantum state. 
  Then, there exists a projector $\Pi_{AC}$ and a state 
  $\rhot_{ABC} = (\Pi_{AC} \otimes \id_B) \rho_{ABC} (\Pi_{AC} \otimes \id_B)$ 
  such that $\rhot_{ABC} \approx^\eps\! \rho_{ABC}$ and
  \begin{align}
    h_{\min}(\rhot_{AB} \| \id_A \otimes \rho_B) \geq H_{\min}(A|B)_{\rho} - 
    \log_2 \, (2/\eps^2) \nonumber.
  \end{align}
\end{lemma}

The next lemma provides a similar upper bound for the smooth min-entropy.

\begin{lemma}
  \label{lm:rel-smooth-bound}
  Let $\eps > 0, \eps' \geq 0$ and $\rho_{AB}$ a quantum state. 
  Then, there exists a state $\rhob_{AB}$ with $P(\rhob_{AB}, \rho_{AB}) \leq \eps + 2\eps'$ such that
  \begin{align}
    h_{\min}(\rhob_{AB} \| \id_A \otimes \rho_B) \geq 
      H_{\min}^{\eps'}(A|B)_{\rho} - \log_2 \, (2/\eps^2) \nonumber.
  \end{align}  
\end{lemma}

\begin{proof}
  Let $\rho_{ABC}$ and $\rhoh_{ABC} \approx^{\eps'}\! \rho_{ABC}$ 
  be pure states such that 
  $H_{\min}^{\eps'}(A|B)_{\rho} = H_{\min}(A|B)_{\rhoh}$. We apply 
  Lemma~\ref{lm:rel-min-bound} to this state to get
  $h_{\min}(\rhot_{AB} \| \id_A \otimes \rhoh_B) 
      \geq h_{\min}^{\eps'}(A|B)_{\rho} - \log_2 \, (2/\eps^2)$,
  where $\ket{\rhot_{ABC}} = (\Pi_{AC} \otimes \id_B) \ket{\rhoh_{ABC}}$ and 
  $\rhot_{ABC} \approx^{\eps}\! \rhoh_{ABC}$.
  Using Eq.~\eqref{eq:pd-ext}, we define the operator $X_B$ with the 
  property $X_B \rhoh_B X_B^{\dagger} = \rho_B$; hence $X_B \rhoh_{ABC} X_B^\dagger
  \approx^{\eps'} \rhoh_{ABC}$. 
  
  Applying this to the defining operator inequality of the relative entropy above leads to
  \begin{align}
    \rhot_{AB} \preceq 2^{-\lambda} \id_A \otimes \rhoh_B \implies 
    \underbrace{X_B \rhot_{AB} X_B^{\dagger}}_{=: \rhob_{AB}} \preceq 2^{-\lambda} \id_A 
    \otimes \rho_B \nonumber
  \end{align}
  and, thus, $h_{\min}(\rhot_{AB} \| \id_A \otimes \rhoh_B) \leq 
  h_{\min}(\rhob_{AB} \| \id_A \otimes \rho_B)$. Furthermore, $\rhob_{AB}$ is 
  sub-normalized since $\tr(\rhob_B) = \tr(X_B \rhot_B X_B^\dagger) \leq 
  \tr(X_B \rhoh_B X_B^\dagger) = \tr(\rho_B) \leq 1$.
  Hence, it remains to bound $P(\rhob_{AB}, \rho_{AB}) \leq P(\rhob_{AB}, \rhot_{AB}) 
  + P(\rhot_{AB}, \rhoh_{AB}) + P(\rhoh_{AB}, \rho_{AB}) \leq P(\rhob_{AB}, \rhot_{AB}) + \eps + \eps'$. We have
  \begin{align} 
    P(\rhob_{AB}, \rhot_{AB}) &= P \big( (X_B \otimes \Pi_{AC}) \,\rhoh_{ABC}\, 
      (X_B^\dagger \otimes \Pi_{AC}),\, (\Pi_{AC} \otimes \id_B) \,\rhoh_{ABC}\, (\Pi_{AC} \otimes \id_B) \big) 
      \nonumber\\
      &\leq P(X_B \,\rhoh_B\, X_B^\dagger, \rhoh_B) \leq \eps' \nonumber,
  \end{align}
  where we used the monotonicity of the purified distance under projections.
\end{proof}

\subsection{Uncertainty of Two Consecutive Measurements}

We prove a more general result that implies Theorem~\ref{th:ucrmin}.
For this purpose, we consider two consecutive measurements applied to the $A$ system and a state $\rho_{ABC}$:  a projective measurement, $\bK = \{ P_A^k \}_k$, followed by either one of two POVMs, $\bX = \{ M_A^x \}_x$ or $\bY = \{ N_A^y \}_y$.
More precisely, we are interested in the post measurement states
\begin{align}
  \rho_{XKB} &= \sum_{x,k} \proj{x} \otimes \proj{k} \otimes \tr_{AC} \big( \big( P_A^k M_A^x P_A^k \otimes \id_{BC} \big) \rho_{ABC} \big)
    \quad \textrm{and} \label{eq:ucr/pm-xk}\\
  \rho_{YKC} &= \sum_{y,k} \proj{y} \otimes \proj{k} \otimes \tr_{AB} \big( \big( P_A^k N_A^y P_A^k \otimes \id_{BC}\big) \rho_{ABC} \big) 
    \label{eq:ucr/pm-yk}\,.
\end{align}

\begin{proposition}
  \label{pr:ucr}
  Let $\rho_{ABC}$ be a tripartite quantum state, let $\eps \geq 0$ and let 
  $\bar{\eps} > 0$.  
  Moreover, let $\bK = \{ P_A^k \}_k $ be a projective measurement and 
  $\bX = \{ M_A^x \}_x$ and $\bY = \{ N_A^y \}_y$ 
  be two POVMs on $A$. Then, the post measurement states~\eqref{eq:ucr/pm-xk} and~\eqref{eq:ucr/pm-yk} satisfy
  \begin{align}
    \label{eq:ucrthm}
    H^{\eps+2\bar{\eps}}_{\min}(X|BK)_{\rho} + H^{\eps}_{\max}(Y|CK)_{\rho}
      \geq - \log_2 c^*_{\bK}(\rho_A, \bX, \bY) - \log_2\, (2/\bar{\eps}^2)  \,,
  \end{align}
  where $c^*_{\bK}(\rho_A, \bX, \bY) := \sum_k \tr(P_A^k \rho_A) \max_{x} \big\| \sum_y P_A^k N_A^y P_A^k \cdot P_A^k M_A^x P_A^k \cdot P_A^k N_A^y P_A^k\, \big\|$.
\end{proposition}

\begin{proof}
  We first prove the statement for pure $\rho_{ABC}$. Then, for mixed states, we consider a purification $\rho_{ABCE}$ of 
  $\rho_{ABC}$, for which the theorem holds and take the partial trace over $E$. 
  As this cannot decrease the smooth entropies~\eqref{eq:data-proc}, the 
  generalization follows.

  We consider the Stinespring dilation of the joint 
  measurement of $\bX$ and $\bK$, denoted $U$, which coherently stores the measurement
  outcome of $\bX$ in registers $X$ and $X'$ and the measurement outcome of 
  $\bK$ in $K$ and $K'$, i.e.\
  $U := \sum_{x,k} \ket{x}_{X} \otimes \ket{x}_{X'} \otimes \ket{k}_{K} \otimes \ket{k}_{K'} 
  \otimes \sqrt{M_A^x}\, P_A^k$.
  Similarly, we introduce the Stinespring dilation of the joint measurement of 
  $\bY$ and $\bK$, and the partial isometry $W := U V^\dagger$ which, using 
  $P_A^k P^{k'} = \delta_{kk'} P_A^k$, evaluates to
  \begin{align}
    W = \sum_{x,y,k} \ket{x}\!\bra{y} \otimes \ket{x}\!\bra{y} \otimes \proj{k} \otimes \proj{k} 
      \otimes \sqrt{M_A^{x\phantom{y}\!\!\!}} P_A^k \sqrt{N_A^y} \label{eq:ucr/gp0}\,.
  \end{align}
  
  These isometries allow us to introduce the states $\rho_{AXX'KK'BC} = U \rho_{ABC} U^\dagger$
  and, analogously, $\rho_{AYY'KK'BC} = V \rho_{ABC} V^\dagger$, whose marginals correspond to the 
  post measurement states $\rho_{XKB}$ and $\rho_{YKC}$ of~\eqref{eq:ucr/pm-xk} 
  and~\eqref{eq:ucr/pm-yk}, respectively.
  
  The proof now proceeds in several steps. First, we reformulate the statement of the theorem 
  in terms of smooth min-entropies using the definition of the smooth max-entropy.
  Then, we use Lemma~\ref{lm:rel-smooth-bound}
  to find an upper bound 
  on one of the entropies in terms of a relative entropy of the state and its marginal. The
  structure of the marginal can then be used to extract $c_\bK^*$.
  
  Due to the duality~\cite{tomamichel09} between smooth min- and max-entropy, the statement of the proposition 
  is equivalent to
 $H_{\min}^{2\eps+\bar{\eps}}(X|KB)_{\rho} \geq H_{\min}^{\eps}(Y|AY'K'B)_{\rho} 
      - \log_2 c_{\bK}^* - 
      \log_2\, (2/\bar{\eps}^2)$.
  Applying Lemma~\ref{lm:rel-smooth-bound}, we introduce a state $\rhot \approx^{2\eps+\bar{\eps}}\! \rho$ such that
  \begin{align*}
    h_{\min}(\rhot_{AYY'K'B} \| \id_Y \otimes \rho_{AY'K'B}) 
      \geq H_{\min}^{\eps}(Y|AY'K'B)_{\rho} - \log_2\, (2/\bar{\eps}^2) \,.
  \end{align*}
  
  Next, we use the monotonicity of $h_{\min}$ under trace-preserving completely positive maps to measure the 
  $K'$ system. More precisely, we apply the map $\mathcal{M}: \rho \mapsto \sum_k \proj{k}_{K'} \rho \proj{k}_{K'}$ to both arguments in $h_{\min}$ above.
  This has no effect on $\rho_{AY'K'B}$, which is classical on $K'$ by definition. 
  Using the state $\rhob_{AYY'K'B} = \mathcal{M}[\rhot_{AYY'K'B}]$, we thus have
  \begin{align}
    \underbrace{h_{\min}(\rhob_{AYY'K'B}\|\id_Y \otimes \rho_{AY'K'B})}_{=:\ \lambda} 
      &\geq H_{\min}^{\eps}(Y|AY'K'B)_{\rho} - \log_2\, (2/\bar{\eps}^2) 
      \label{eq:ucr/gp2}\,.
  \end{align}
  Moreover, the purified distance satisfies $P(\rhob, \rho) \leq P(\rhot, \rho) \leq 2\eps + \bar{\eps}$.

  From the definition of 
  $h_{\min}$, we get
  \begin{align}
    \rhob_{AYY'K'B} \preceq 2^{-\lambda}\, \id_Y \otimes \rho_{AY'K'B} \,, \label{eq:ucr/gp3} 
  \end{align}
  where we employed the marginal state $\rho_{AY'K'B} = \tr_{YK'}(V \rho_{AB} V^{\dagger}) = \sum_{y,k}\! \sqrt{N_A^y} P_A^k \rho_{AB} P_A^{k}\! \sqrt{N_A^y} \otimes 
      \proj{k} \otimes \proj{y}$.
  Taking the tensor product with $\id_{K}$ on both sides of~\eqref{eq:ucr/gp3}, conjugating the resulting inequality with $W$ and taking the partial trace over 
  $A$, $Y'$ and $K'$ leads to
  \begin{align}
    \underbrace{\tr_{AX'K'} \big( W (\rhob_{AYY'K'B} \otimes \id_{K}) W^{\dagger} \big)}_{=:\ \taub_{XKB}}
      \preceq 2^{-\lambda}\, \tr_{AX'K} \big( W (\id_{YK} \otimes \rho_{AY'K'B}) W^{\dagger} \big) \label{eq:ucr/gp5}\,.
  \end{align}
  We evaluate the trace term on the rhs.\ of~\eqref{eq:ucr/gp5} to get
  \begin{align}
    &\tr_{AX'K'} \big( W (\id_{YK} \otimes \rho_{AY'K'B}) W^{\dagger} \big) \nonumber\\
    &\quad\!\! = \sum_{x,y,k} \proj{x} \otimes \proj{k} \otimes 
      \brakket{yk}{\tr_{A} \big( \sqrt{M_A^x} P_A^k\! \sqrt{N_A^y} \rho_{AY'K'B} \sqrt{N_A^y} P_A^k\! \sqrt{M_A^x} \big)}{yk} \nonumber\\
    &\quad\!\! = \sum_{x} \proj{x} \otimes \sum_k \proj{k} \otimes 
      \tr_{A} \Big( \sum_y \sqrt{M_A^x} P_A^k\! N_A^y P_A^k \rho_{AB} P_A^k N_A^y P_A^k\! \sqrt{M_A^x} \Big) \nonumber\\
    &\quad\!\! \preceq \id_X \otimes \underbrace{\sum_k \proj{k} \otimes 
      \max_x \big\| \sum_y P_A^k N_A^y P_A^k M_A^x P_A^k N_A^y P_A^k \big\| 
      \tr_{A} (P_A^k \rho_{AB})}_{=:\ \omegat_{KB}} \label{eq:ucr/gp25}
  \end{align}
  We used Lemma~\ref{lm:pt-norm-bound} to arrive at~\eqref{eq:ucr/gp25}.
  Note that $\tr(\omegat_{KB}) = c_{\bK}^*$; hence, we choose $\omega_{KB} = \omegat_{KB}/c_{\bK}^*$ and employ~\eqref{eq:ucr/gp5} to find a lower bound on $h_{\min}(\taub_{XKB}\|\id_X \otimes \omega_{KB})$ in
  terms of $\lambda$ and $c_{\bK}^*$, i.e.
  \begin{align}
    h_{\min}(\taub_{XKB}\|\id_X \otimes \omega_{KB}) &\geq \lambda - \log_2 c_{\bK}^* \nonumber\\
      &\geq H_{\min}^{\eps}(Y|AY'K'B)_{\rho} -
      \log_2 c_{\bK}^* - \log_2\, (2/\bar{\eps}^2) \label{eq:ucr/gp6}\,.
  \end{align}
  We have
$P(\taub_{XKB}, \rho_{XKB}) = P(\rhob_{XKB}, \rho_{XKB}) \leq 2\eps + \bar{\eps}$.
  Therefore, using the definition of the smooth min-entropy, we get
$H_{\min}^{2\eps+\bar{\eps}}(X|KB)_{\rho} \geq 
h_{\min}(\taub_{XKB} \| \id_X \otimes \omega_{KB})$,
  which, substituted into~\eqref{eq:ucr/gp6}, concludes the proof.
\end{proof}

\subsection{Proof of Theorem~\ref{th:ucrmin}}

Theorem~\ref{th:ucrmin} is a corollary of Proposition~\ref{pr:ucr}. 

\begin{proof}[Proof of Theorem~\ref{th:ucrmin}]
Recall that the effective overlap is defined as
\begin{align*}
  c^*(\rho_A, \bX, \bY) = \inf_{U,\, \bX', \bY', \bK'} c_{\bK'}^*(U \rho_A U^{\dagger}, \bX', \bY') \,,
\end{align*}
where the infimum is taken over all embeddings $U$ from $A$ to $A'$, all measurements
$\bX' = \{ M_{A'}^x \}_x$ and $\bY' = \{ N_{A'}^y \}_y$ on $A'$ and all projective measurements $\bK' = \{ P_{A'}^k \}_k$ such that $\sum_k U^{\dagger} P_{A'}^k M_{A'}^x P_{A'}^k U = M_A^x$ and 
$\sum_k U^{\dagger} P_{A'}^k N_{A'}^y P_{A'}^k U = N_A^y$.
Furthermore, for any such $\{U, \bX', \bY', \bK' \}$, Proposition~\ref{pr:ucr} implies that the
post measurement states
\begin{align}
  \tau_{XKB} &= \sum_{x,k} \proj{x} \otimes \proj{k} \otimes \tr_{A'C} \big( \big( P_{A'}^k M_{A'}^x P_{A'}^k \otimes \id_{BC} \big) U \rho_{ABC} U^{\dagger} \big)
    \quad \textrm{and} \label{mm1} \\
  \tau_{YKC} &= \sum_{y,k} \proj{y} \otimes \proj{k} \otimes \tr_{A'B} \big( \big( P_{A'}^k N_{A'}^y P_{A'}^k \otimes \id_{BC}\big) U \rho_{ABC} U^{\dagger} \big) \label{mm2} \,.
\end{align}
satisfy
\begin{align*}
    H^{\eps+2\bar{\eps}}_{\min}(X|B)_{\tau} + H^{\eps}_{\max}(Y|C)_{\tau} &\geq
    H^{\eps+2\bar{\eps}}_{\min}(X|BK)_{\tau} + H^{\eps}_{\max}(Y|CK)_{\tau} \\
      &\geq - \log_2 c^*_{\bK'}(U \rho_A U^{\dagger}, \bX', \bY') - \log_2\, (2/\bar{\eps}^2)  \,,
\end{align*}
where we also employed the data-processing inequality of the smooth min- and max-entropies~\eqref{eq:data-proc} to trace out the $K$ system.
Furthermore, the marginal states of~\eqref{mm1} and~\eqref{mm2} without $K$ correspond to the post measurement states when measuring $\bX$
and $\bY$ on $\rho$, namely
\begin{align*}
  \tr_K(\tau_{XKB}) &= \sum_x \proj{x} \otimes \tr_{AC}\Big( \Big( \sum_k U^{\dagger} P_{A'}^k M_{A'}^x P_{A'}^k U \Big) \rho_{ABC} \Big) \\
  &= \sum_x \proj{x} \otimes \tr_{AC} \big( M_A^x\, \rho_{ABC} \big) = \cM_{\bX}[\rho_{AB}]
\end{align*}
and, similarly, $\tr_K(\tau_{YKC}) = \cM_{\bY}[\rho_{AC}]$. This implies that the uncertainty relation holds
for each candidate in the minimization and, thus, also for its infimum. (The last argument implicitly uses the continuity of the function $-\log_2$.) This concludes the proof.
\end{proof}

\subsection{Proof of Theorem~\ref{th:ucr}}

Theorem~\ref{th:ucr} follows as a corollary of Theorem~\ref{th:ucrmin} and the entropic
asymptotic equipartition~\eqref{eq:aep}. 

\begin{proof}[Proof of Theorem~\ref{th:ucr}]
We apply Theorem~\ref{th:ucrmin}
to the state $\rho_{ABC}^n = \rho_{ABC}^{\otimes n}$ and use the measurements $\bX^n$ and $\bY^n$,
which measure $\bX$ and $\bY$ on each of the $n$ copies, respectively.
It is easy to verify that $c^*(\rho_A^n, \bX^n, \bY^n) \leq 
c^*(\rho_A, \bX, \bY)^n$ in this case. Theorem~\ref{th:ucrmin} applied
to this situation thus yields
\begin{align*}
    \frac{1}{n} H_{\min}^{\eps+2\bar{\eps}}(X^n|B^n)_{\rho} + \frac{1}{n} H_{\max}^{\eps}(Y^n|C^n)_{\rho} \geq - \log_2 c^*(\rho_A, \bX, \bY) - \frac{1}{n} \log_2\, (2/\bar{\eps}^2) \,.
\end{align*}
Finally, taking the limit $n \to \infty$ and employing~\eqref{eq:aep} immediately proves Theorem~\ref{th:ucr}. 
\end{proof}

\section{Proof of Relation to Nonlocality}
\label{sc:overlap}

\subsection{Preliminaries}

Projective measurements with binary outcomes can be described compactly as 
an \emph{observable} $O = M_0 - M_1$ with spectrum in $\{1, -1\}$, i.e.\ $O^2 = \id$.
Tsirelson~\cite{tsirelson80} related the correlations which can be achieved when
measuring quantum 
systems to the existence of unit vectors 
in a real vector space. 
Namely, Tsirelson's result states that for any set of observables 
$O_1,\dotsc,O_n$ and $Q_1,\ldots,Q_n$ with eigenvalues
in the interval $[-1,1]$ and any bipartite pure state 
$\proj{\psi}$ there exist real unit vectors $x_1, \dotsc, 
x_n, y_1,\dotsc, y_n \in \bR^{2n}$ s.t.\
\begin{align}
  \brakket{\psi}{O_i \otimes Q_j}{\psi} &= x_i^\intercal \cdot y_j 
  \label{eq:tsirelson-observables}
\end{align}
for all $i,j \in [n]$. Conversely, if there exist such real unit vectors
$x_i$ and $y_j$, it is 
possible to find sets of observables $O_i$ on $\mathcal{H}$ and $Q_i$ on 
$\mathcal{H}'$ with eigenvalues $\pm 1$ and $\dim \mathcal{H} = \dim \mathcal{H}' = n$ such 
that \eqref{eq:tsirelson-observables} holds with  $\ket{\psi}$ a maximally entangled state. 

As shown by Wehner~\cite{wehner06}, this implies that the maximal CHSH value reachable 
by a quantum system can be calculated using a
\emph{semidefinite program} (SDP), more precisely, an optimization problem of the 
form $\texttt{max:} \tr (B G)$, $\texttt{subject to:} \tr (E_i G) = e_i$ for all $i$, and $G \succeq 0$. 
Here, $\{ E_i, e_i \}_i$ is a set of linear constraints and $G$ 
is the variable to be optimized over (we refer to e.g.~\cite{boyd04} for details on 
semidefinite programming). The reason for this is, that a (real symmetric) 
matrix $G$ is positive semidefinite if and only if it can be expressed as 
 $G = B^\intercal B$, i.e., its entries are the inner product of the vectors 
 representing the columns of $B$. 

For example, for the case of two inputs and outputs, the
correlations can be arranged in a $4\times 4$ matrix
$G = ( g_{ij} ) \text{ with } g_{ij} := x_i^\intercal \cdot x_j$.
Conversely, any $4\times 4$ positive semidefinite matrix with diagonal entries
equal to $1$ can be seen as an arrangement of this sort, since 
$G = B^\intercal B$ where $B = ( x_1, x_2, y_1,
y_2 )$.
The expected CHSH value, $\beta$, of a certain setup between two parties
can be calculated from this matrix $G$ using
\begin{align*}
  \beta(\ket{\psi}, O_1, O_2, Q_1, Q_2) = \brakket{\psi}{O_1 \!\otimes\! Q_1 \!+\! O_1 \!\otimes\! Q_2 
    \!+\! O_2 \!\otimes\! Q_1 \!-\! O_2 \!\otimes\! Q_2}{\psi} = \tr ( W G ), 
\end{align*}
where $G$ is defined as above and
\begin{align*}
 \quad W := \frac{1}{2}
  \left(
     \begin{array}{cccc}
      0 & 0 & 1 & 1 \\
      0 & 0 & 1 & -1 \\
      1 & 1 & 0 & 0 \\
      1 & -1 & 0 & 0
    \end{array}
    \right).
\end{align*}

\subsection{Generalization of Tsirelson's Results}

We are here concerned with extending Tsirelson's relation between symmetric matrices 
and bipartite measurements of the previous section to the case
where the overlap of the local observables is restricted. For this purpose, we first define an effective
overlap of two observables.
\begin{definition}
Let $O_1$, $O_2$ be observables on $\mathcal{H}$ with binary spectrum $\{-1, 1\}$
and let $\rho$ be a density operator acting on $\mathcal{H}$. The \emph{effective overlap 
between the observables $O_1$ and $O_2$ on $\rho$} is 
\begin{align}
  \nonumber \gamma^*(\rho, O_1, O_2) &:=
  \frac{1}{4} \tr \big( \rho 
    ( O_1 + O_2 )^2 \big) .
\end{align}
\end{definition}
\noindent We will later make a connection between this quantity and the effective overlap of POVMs, $c^*$.

The following Lemma is an extension of Tsirelson's~\cite{tsirelson80} original
relation in the form used in~\cite{wehner06}.

\begin{lemma}
\label{lem:tsirelsonext}
Let $\rho_{AB}$ be a bipartite quantum state.
Furthermore, let $O_1,O_2,\ldots,O_n$ be observables with binary spectrum
$\{-1,1\}$ on $A$ and let $Q_1,Q_2,\ldots,Q_m$ be observables with
binary spectrum $\{-1,1\}$ on $B$. 
Then, there exists a real positive semidefinite 
$(n + m) \times (n + m)$ matrix $G$ such that, for all $i, i' \in [n]$, $j, j'
\in [m]$,
\begin{align*}
  (G)_{i(n+j)} &= (G)_{(n+j)i} = \tr \big( (O_i \otimes Q_j ) \rho_{AB} \big)
 \\
 (G)_{ii^\prime} &= 2 \gamma^*(\rho_A, O_i, O_{i^\prime}) - 1 \\
  (G)_{(n+j)(n+j^\prime)} &= 2 \gamma^*(\rho_B, Q_j, Q_{j^\prime} ) - 1
\end{align*}
\end{lemma}

\begin{proof}
To prove the statement, we construct the matrix $G$ for given $\rho_{AB}$ and
observables $O_i$ and $Q_j$. Let $\ket{\psi}$ be a purification of $\rho_{AB}$
on an auxiliary system $C$. 
Then, we define vectors for all $i \in [n], j \in [m]$: $x_i := (O_i \otimes \id_B
\otimes \id_C) \ket{\psi}$ and $x_{n+j} := (\id_A \otimes Q_j \otimes \id_C)
\ket{\psi}$. The $(n+m) \times (n+m)$ matrix $\bar{G}$ given by the inner
products, i.e.\ $(\bar{G})_{kk^\prime} = {x_k}^\dagger x_{k^\prime}$, is Hermitian and
positive semidefinite by construction. Finally, $G = (\bar{G} + \bar{G}^\intercal)/2$ is
positive semidefinite, real and symmetric. 

It remains to check that the correlations
agree.
First, note that
\begin{align*} 
  (G)_{i(n+j)} = (\bar{G})_{i(n+j)} = \brakket{\psi}{O_i
  \otimes Q_j \otimes \id_C}{\psi} = \tr \big( (O_i \otimes Q_j ) \rho_{AB} \big) \,.  
\end{align*}
Moreover, the local terms on $A$ evaluate to
\begin{align*}
\nonumber (G)_{ii'} &= \frac{1}{2} \brakket{\psi}{(O_i O_{i'} \otimes
\id_{BC})}{\psi} + \frac{1}{2} \brakket{\psi}{(O_{i'} O_i \otimes
\id_{BC})}{\psi}  \\
\nonumber &=\frac{1}{2}\tr \big( \rho_A (O_i O_{i'} + O_{i'} O_i) \big) 
  =2 \gamma^*(\rho_A, O_i, O_{i'}) -1 
\end{align*}
and similarly on $B$ with $(G)_{(n+j)(n+j)}$.
\end{proof}

The converse is also true, for every matrix $G$ satisfying above properties, there exists a 
physical realization. This corresponds to the converse of 
Tsirelson's theorem~\cite{tsirelson80,tsirelson93} 
(see also~\cite{wehnerphd} for a detailed explanation). 
\begin{lemma}
\label{lem:tsirelsonext2}
Let $G$ be a real positive semidefinite $(n + m) \times (n + m)$ matrix with $(G)_{ii}=1$.
Then there exists a quantum state $\rho_{AB}$, 
observables $O_1,O_2,\ldots,O_n$ with binary spectrum
$\{-1,1\}$ on $A$  and observables $Q_1,Q_2,\ldots,Q_m$  with
binary spectrum $\{-1,1\}$ on $B$, such that, for all $i, i' \in [n]$, $j, j'
\in [m]$, it holds that
\begin{align*}
  \tr \big( (O_i \otimes Q_j ) \rho_{AB} \big) &= (G)_{i(n+j)} \\
  2 \gamma^*(\rho_A, O_i, O_{i'}) - 1 
    &= (G)_{ii'} \\
 2 \gamma^*(\rho_B, Q_j, Q_{j'}) - 1 &=   (G)_{(n+j)(n+j')} 
\end{align*}
\end{lemma}

\begin{proof}
Let $d = n + m$ and $\{ x_k \}$, $k \in \{1, \ldots, d \}$ be a set of real vectors of dimension 
$d$ such that $(G)_{kk'} = x_k^{\intercal} x_{k'}$.
Moreover, take $\rho_{AB} = \proj{\psi}$, where $\ket{\psi} = \sqrt{d}{}^{-1} \sum_k \ket{k} \ket{k}$ is 
the maximally entangled state, $O_i = \sum_{\ell} (x_{i})_\ell^{\phantom{\intercal}} \Gamma_{\ell}^\intercal$ and $Q_j = \sum_{\ell} (x_{n+j})_{\ell} \Gamma_{\ell}$
 where $\Gamma_\ell$ are generators of the Clifford algebra in dimension $n+m$, i.e., 
 $\{\Gamma_\ell,\Gamma_{\ell'}\} = 2 \delta_{\ell\ell'} \id$. Using the fact that 
 $\Gamma_{\ell}$ are anti-commuting, it is now straight forward to verify that 
the $O_i$ and $Q_j$ have spectrum in $\{ -1, 1\}$ since
\begin{align*}
  O_i O_{i'} &= \bigg( \sum_{\ell} (x_{i})_{\ell}^{\phantom{\intercal}} \Gamma_{\ell}^\intercal \bigg)
  \bigg( \sum_{\ell'} (x_{i'})_{\ell'}^{\phantom{\intercal}} \Gamma_{\ell'}^\intercal \bigg)
  = \frac{1}{2} \sum_{\ell, \ell'}  (x_{i})_{\ell} (x_{i'})_{\ell'} \{ \Gamma_{\ell}, 
  \Gamma_{\ell'} \}^\intercal  = x_i^\intercal x_{i'} \id .
\end{align*}
Thus, $2\gamma^*(\rho_A, Q_i, Q_{i'}) - 1 = \frac{1}{2} \tr(\rho_A \{ O_i, O_{i'} \}) = (G)_{ii'}$ and similarly for $(G)_{(n+j)(n+j')}$.
Finally,
\begin{align*}
  \brakket{\psi}{O_i\otimes Q_j}{\psi}&= 
  \frac{1}{d} \sum_{\ell, \ell'} (x_i)_{\ell} (x_{n+j})_{\ell'} 
  \bigg(
  \sum_{k,k'}\bra{k}\bra{k}\Gamma_{\ell}^\intercal \otimes  \Gamma_{\ell'} \ket{k'}\ket{k'} \bigg) \\
  \nonumber &= 
  \frac{1}{d} \sum_{\ell, \ell'} (x_i)_{\ell} (x_{n+j})_{\ell'} \tr\left(\Gamma_{\ell}\Gamma_{\ell'} \right)
  = \sum_{\ell, \ell'} (x_i)_{\ell} (x_{n+j})_{\ell'}  \delta_{\ell\ell'} = (G)_{i(n+j)} \,. \qedhere
\end{align*}
\end{proof}

\subsection{Two Binary Measurements}

Next, we restrict our attention to the case where two parties, Alice and David, each have two observables at their disposal. 
The measurement setup can in this case be described by the set $\{ \ket{\psi}, O_1, O_2, Q_1, Q_2 \}$. We define 
the following family of semidefinite programs, which calculate the maximal CHSH value, $\beta_{\max}(\gamma^*)$, 
that can be achieved with a setup for which the effective overlap of Alice's observables satisfies 
$\gamma^*(\rho_A, O_1, O_2) = \gamma^*$.
The SDP for $\beta_{\max}(\gamma^*)$ is given by
\begin{align}
  \texttt{maximize:} \quad &
    \tr \big( W G \big) \nonumber\\
  \texttt{subject to:} \quad & G \succeq 0, \nonumber\\
      &(G)_{ii} = 1\ \forall i\ \text{ and} \nonumber\\
      &(G)_{12} = (G)_{21} = 2\gamma^* - 1 \,. \label{eq:SDP}
\end{align}
Note that, since every physical setup has a corresponding matrix $G$ due to Lemma~\ref{lem:tsirelsonext}, the 
maximization is done over all physical setups that satisfy the constraint on the effective overlap. On the 
other hand, Lemma~\ref{lem:tsirelsonext2} tells us that there exists a physical setup\,|\,corresponding to 
the optimal matrix $G^*$\,|\,that achieves any $\beta_{\max} = \tr (W G^*)$.
Note, however, that this does not imply that every setup with a given $\gamma^*$ can be used to reach $\beta_{\max}(\gamma^*)$. 

The function $\beta_{\max}(\gamma^*)$ has a nice analytical form, which was conjectured by M.~Horodecki~\cite{horodecki-pc11} for the two qubit case. Alternatively, it is possible to derive a statement of this type~\cite{lim-pc12} using a result of Seevink and Uffink~\cite{seevinck07}, which bounds the maximal CHSH value in terms of the angle between local qubit measurements.
\begin{lemma}
  \label{lem:horodecki}
  The maximal CHSH value $\beta_{\max}$ that can be achieved by a setup $\{\rho_{AT}, O_1, O_2, Q_1, Q_2 \}$ that has 
  a effective overlap $\gamma^*(\rho_A, O_1, O_2) = \gamma^*$ is given by
  \begin{align*}
    \beta_{\max}(\gamma^*) = 2 \big( \sqrt{\gamma^*} + \sqrt{1-\gamma^*} \big) \,.
  \end{align*}
\end{lemma}

\begin{proof}
The solution is given by the SDP~\eqref{eq:SDP} and it remains to find feasible solutions 
for both the primal and the dual problem in order to find $\beta_{\max}$.
We first construct a primal feasible solution $G^*$ for the SDP~\eqref{eq:SDP}. We have,
\begin{align*}
  & \qquad \beta_{\max}(\gamma^*) \geq \tr \big( W G^* \big) = 
    2 \big( \sqrt{\gamma^*} + \sqrt{1-\gamma^*} \big)\,, \text{ where } \\
  \nonumber  &
    G^* := 
    \left(\begin{array}{cccc}
      1 & 2\gamma^*-1 & \sqrt{\gamma^*} & \sqrt{1-\gamma^*} \\
      2\gamma^*-1 & 1 & \sqrt{\gamma^*} & -\sqrt{1-\gamma^*} \\
      \sqrt{\gamma^*} & \sqrt{\gamma^*} & 1 & 0 \\
      \sqrt{1-\gamma^*} & -\sqrt{1-\gamma^*} & 0 & 1
    \end{array}\right) \succeq 0\,, \quad \textrm{for all $\gamma^* \in [0, 1]$.}
\end{align*}

To find an upper bound on $\beta_{\max}$, we consider the dual SDP, which is
\begin{align}
 \texttt{minimize:} \quad &
   \Gamma_{11}+\Gamma_{22}+\Gamma_{33}+\Gamma_{44} + (2\gamma^* - 1) (\Gamma_{12} + \Gamma_{21}) \nonumber\\
\nonumber  \texttt{subject to:} \quad &   \Gamma=
   \left(\begin{array}{cccc}
     \Gamma_{11} &  \Gamma_{12} & 0 & 0 \\
     \Gamma_{21} &  \Gamma_{22} & 0 & 0 \\
    0 & 0 &  \Gamma_{33} & 0 \\
    0 & 0 & 0 &  \Gamma_{44}
   \end{array}\right)
   \succeq W \,.
  \end{align}
A feasible solution, $\Gamma^* \succeq W$, is 
\begin{align*}
  \Gamma^* &:=  \left(\begin{array}{cccc}
    \frac{1}{4} \big(\frac{1}{\sqrt{\gamma^*}} + \frac{1}{\sqrt{1-\gamma^*}}\big) & 
    \frac{1}{4} \big(\frac{1}{\sqrt{\gamma^*}} - \frac{1}{\sqrt{1-\gamma^*}}\big) & 0 & 0 \\
    \frac{1}{4} \big(\frac{1}{\sqrt{\gamma^*}} - \frac{1}{\sqrt{1-\gamma^*}}\big) & 
    \frac{1}{4} \big(\frac{1}{\sqrt{\gamma^*}} + \frac{1}{\sqrt{1-\gamma^*}}\big) & 0 & 0 \\
    0 & 0 & \sqrt{\gamma^*} & 0 \\
    0 & 0 & 0 & \sqrt{1-\gamma^*}
   \end{array}\right) \,,
\end{align*}
Thus, due to weak duality of semidefinite programming, it holds that $\beta_{\max}(\gamma^*) \leq \tr (\Gamma^*) + (2\gamma^*-1) (\Gamma^*_{12} + \Gamma^*_{21}) = 2 \big(  \sqrt{\gamma^*} + \sqrt{1-\gamma^*} \big)$, which concludes the proof.
\end{proof}

\subsection{Proof of Theorem~\ref{th:maxoverlap}}

We will need a pivotal result due to Jordan~\cite{jordan1875} (see also~\cite{tsirelson93, masanes06}).
\begin{lemma}[Jordan's Lemma]
  \label{lm:jordan}
  Let $\bX = \{ M^0, M^1 \}$ and $\bY = \{ N^0, N^1 \}$ 
  be two projective measurements with binary outcomes. Then, there exists a projective
  measurement $\bK = \{ P^k \}_k$ that commutes with both 
  $\bX$
  and $\bY$ such that the $P^k$ project on subspaces of dimension 
  at most $2$.
\end{lemma}

Let now $\bK$ be such a projective measurement for which we additionally require that the rank of the $P^k$ is minimal. It is easy to verify that this measurement has the property that the projectors $P^k M^x P^k$ and $P^k M^x P^k$ either vanish or are rank-$1$ projectors. (If, for example, $P^k M^0 P^k$ is not rank-$1$, it must either be $P^k$ or vanish. However, this implies that $P^k M^1 P^k$ also either vanishes or equals $P^k$ and, thus, measuring further
in the basis induced by $P^k M^y P^k$ will reduce the dimension of $\bK$.)
Hence, the projectors can be written in the form $\proj{\xi_k^x} = P^k N^y P^k$ and $\proj{\zeta_k^y} = P^k N^y P^k$, where $\ket{\xi_k^x}$ and $\ket{\zeta_k^x}$ are allowed to be the zero vector.

It remains to relate the effective overlap of observables, $\gamma^*$, to the effective overlap of two POVMs, $c^*$. This is done in the following proposition, from which Theorem~\ref{th:maxoverlap} directly follows. 

\begin{proposition}
  \label{pr:cstarbeta}
  For any measurement setup $\{ \rho_A, \bX, \bY \}$, it holds that
  \begin{align}
    c^*(\rho_A, \bX, \bY) \leq \frac{1}{2} + \frac{\beta}{8}\sqrt{8 - \beta^2} \,,
    \label{eq:cstarbeta}
  \end{align}
  where $\beta = \beta(\rho_{AD}, \bX, \bY, \bR, \bS)$ 
  for any extension $\rho_{AD}$ with $\rho_A = \tr_D(\rho_{AD})$ and for any two
  binary POVMs $\bR$ and $\bS$ on $D$. 
\end{proposition}

\begin{proof}
  It is sufficient to consider projective measurements and pure states as, due to Neumark's dilation 
  theorem and the definition of the effective overlap, there exist projective measurements $\bX'$, $\bY'$, 
  $\bR'$, $\bS'$ and an embedded state $\rho_{A'D'}$ such that
  \begin{align}
      \beta(\rho_{AD}, \bX, \bY, \bR, \bS) = \beta(\rho_{A'D'}, \bX', \bY', \bR', \bS')
    \qquad \textrm{and} \qquad 
      c^*(\rho_{A}, \bX, \bY) \leq \min_{\bK'}\, c_{\bK'}^*(\rho_{A'}, \bX', \bY') \,,
      \label{eq:xxx1}
  \end{align}
  where $\bK' = \{ P_{A'}^k \}_k$ is any projective measurement that commutes with $\bX'$ and $\bY'$.

  According to Lemma~\ref{lm:jordan} and~\eqref{eq:xxx1}, we can thus bound 
  \begin{align}
    c^*(\rho_A, \bX, \bY) 
     \leq \sum_k \tr(P_{A'}^k \rho_{A'}^{\phantom{k}}) \max_x \Big\| \sum_y | \zeta_k^y \rangle\!\langle
     \zeta_k^y | \xi_k^x \rangle\!\langle \xi_k^x | \zeta_k^y \rangle\!\langle \zeta_k^y | \Big\| 
     = \sum_k \tr(P_{A'}^k \rho_{A'}^{\phantom{k}}) \max_{x, y} \big| \!
      \braket{\xi_k^x}{\zeta_k^y} \! \big|^2 \label{eq:maxcc} ,
  \end{align}
  where $P_{A'}^k$ is a decomposition into at most two-dimensional subspaces,
  $\proj{\xi_k^x} = P_{A'}^k M_{A'}^x P_{A'}^k$, and $\proj{\zeta_k^y} = P_{A'}^k N_{A'}^y P_{A'}^k$.
   Now, consider the observables
  \begin{align*}
  \tilde{O}_{A'}^{\bX} = \bigoplus_k \Big( \proj{\xi_k^{x_k}} - \proj{\xi_k^{\bar{x}_k}}\Big) \quad \textrm{and} \quad
  \tilde{O}_{A'}^{\bY} = \bigoplus_k \Big( \proj{\zeta_k^{y_k}} - \proj{\zeta_k^{\bar{y}_k}} \Big) ,   
  \end{align*}
  where $x_k, y_k \in \{0, 1\}$ are the values that maximize the overlap 
  in~\eqref{eq:maxcc} for each value of $k$. Furthermore, $\bar{x}_k = 1 - x_k$ and 
  $\bar{y}_k = 1 - y_k$. Using these observables, it is easy to verify that
  \begin{align}
      \sum_k \tr(P_{A'}^k \rho_{A'}^{\phantom{k}}) \max_{x, y} \big| \!
      \braket{\xi_k^x}{\zeta_k^y} \! \big|^2
      = \gamma^*(\rho_{A'}, \tilde{O}_{A'}^{\bX}, \tilde{O}_{A'}^{\bY}) 
      = \frac{1}{2} + \frac{\beta_{\max}(\gamma^*)}{8} \sqrt{8 - \beta_{\max}(\gamma^*)^2} \nonumber\,,
  \end{align}
  where, in the last step, we used Lemma~\ref{lem:horodecki} and introduce $\beta_{\max}(\gamma^*)$, the maximum CHSH value that can be reached with a bipartite setup that satisifies 
  $\gamma^*(\rho_{A'}, \tilde{O}_{A'}^{\bX}, \tilde{O}_{A'}^{\bY}) = \gamma^*$.
  
  It remains to show that $\beta(\rho_{A'D}, \bX', \bY', \bR', \bS') \leq \beta_{\max}(\gamma^*)$.
  First note that due to the fact that $\bK'$ commutes with $\bX'$ and $\bY'$, we have
  $\beta(\rho_{A'D}, \bX', \bY', \bR', \bS') 
  = \beta(\rho_{A'DK}, \bX', \bY', \bR', \bS')$ where $\rho_{A'DK} = \sum_k \proj{k} \otimes  
  P_{A'}^k \rho_{A'D} P_{A'}^k$. Thus, we can assume without loss of generality 
  that the maximum {CHSH} value is achieved
  with an extension and measurements that potentially depend on the value of $K$.
  Furthermore, we introduce a purification $\ket{\psi}$ of
  $\rho_{A'DK}$ and write 
  \begin{align}
    \beta(\rho_{A'D}, \bX', \bY', \bR', \bS') &\leq \max_{\ket{\psi},\, Q_{D'}^{\bR},\, Q_{D'}^{\bS}}\,
  \beta \big(\ket{\psi}\!, O_{A'}^{\bX}, O_{A'}^{\bY}, Q_{D'}^{\bR}, Q_{D'}^{\bS} \big) \nonumber\\ 
  &= \max_{\ket{\psi},\, Q_{D'}^{\bR},\, Q_{D'}^{\bS}}\, 
  \beta \big(\ket{\psi}\!, \tilde{O}_{A'}^{\bX}, \tilde{O}_{A'}^{\bY}, Q_{D'K}^{\bR}, Q_{D'K}^{\bS} \big) 
  \leq \beta_{\max}(\gamma^*), \label{eq:betabound}
  \end{align}
  where the measurements $\bX$ and $\bY$ are represented as observables
  \begin{align*}
    O_{A'}^{\bX} = M_{A'}^0 - M_{A'}^1 = \bigoplus_k \Big( \proj{\xi_k^0} - \proj{\xi_k^1} \Big) \quad\!
      \textrm{and} \quad\!
    O_{A'}^{\bY} = N_{A'}^0 - N_{A'}^1 = \bigoplus_k \Big( \proj{\zeta_k^0} - \proj{\zeta_k^1} \Big) .  
  \end{align*}
  The equality in~\eqref{eq:betabound} requires some explanation. Note that the observables
  $O$ and $\tilde{O}$ only differ in the way outputs, $0$ or $1$, are labelled for each $k$.
  However, due to the symmetry of the CHSH value, it is easy to 
  verify that David can simulate a $k$-dependent relabeling of Alice's outputs by 
  permuting his inputs and outputs. More precisely, we have
  \begin{align*}
    \beta \big(\ket{\psi}\!, O_{A'}^{\bX}, O_{A'}^{\bY}, Q_{D'}^{\bR}, Q_{D'}^{\bS} \big) 
    = \beta \big(\ket{\psi}\!, \tilde{O}_{A'}^{\bX}, \tilde{O}_{A'}^{\bY}, Q_{D'K}^{\bR}, Q_{D'K}^{\bS} \big)
  \end{align*}
  for the observables $Q_{D'K}^{\bR} = \sum_k \proj{k} \otimes Q_{D'}^{\bR,k}$ and 
  $Q_{D'K}^{\bS} = \sum_k \proj{k} \otimes Q_{D'}^{\bS,k}$, where
  \begin{align*}
    \big\{ Q_{D'}^{\bR,k}, Q_{D'}^{\bS,k} \big\} = (-1)^{x_k}
    \begin{cases} 
      \big\{  Q_{D'}^{\bR}, Q_{D'}^{\bS} \big\} & \text{if } x_k \oplus y_k = 0 \\ 
      \big\{ Q_{D'}^{\bS}, Q_{D'}^{\bR} \big\} & \text{if } x_k \oplus y_k = 1
    \end{cases} \,.
  \end{align*}
  The last inequality in~\eqref{eq:betabound} follows by definition of $\beta_{\max}(\gamma^*)$ 
  and concludes the proof.
\end{proof}

\end{document}